\documentclass[submission]{eptcs}
\providecommand{\event}{GaM 2015} 
\usepackage{breakurl}
\usepackage{amsmath}
\usepackage{amsthm}
\usepackage{amssymb} 
\usepackage{algorithm}
\usepackage{algpseudocode}

\usepackage[pdftex]{graphicx}
\usepackage{paralist}
\usepackage{subfig}

\DeclareMathAlphabet{\mathcal}{OMS}{cmsy}{m}{n}

\usepackage[]{hyperref}
\hypersetup{
  colorlinks=true, linkcolor=black, citecolor=black, urlcolor=black,
}

\graphicspath{ {./figures/} }

\newcommand{\NCP}{conflicting }
\newcommand{\BNCP}{Conflicting }
\newcommand{\dalg}{\mathcal{D}}
\newcommand{\ul}[1]{\underline{#1}}

\newtheorem{defi}{Definition}
\newtheorem{exmp}{Example}
\newtheorem{thm}{Theorem}

\newtheorem{fact}{Fact}

\title{Improved Conflict Detection for Graph Transformation with Attributes}
\author{G{\'e}za Kulcs{\'a}r 
\institute{ 
Technische Universit\"at Darmstadt\\
Real-Time Systems Lab\\ 
Merckstr. 25\\
64283 Darmstadt, Germany\\
}
\and
Frederik Deckwerth\thanks{Supported by CASED (www.cased.de).}
\institute{ 
Technische Universit\"at Darmstadt\\
Real-Time Systems Lab\\ 
Merckstr. 25\\
64283 Darmstadt, Germany\\
}
\and
Malte Lochau
\institute{ 
Technische Universit\"at Darmstadt\\
Real-Time Systems Lab\\ 
Merckstr. 25\\
64283 Darmstadt, Germany\\
}
\and
Gergely Varr\'o
\institute{ 
Technische Universit\"at Darmstadt\\
Real-Time Systems Lab\\ 
Merckstr. 25\\
64283 Darmstadt, Germany\\
}
\and
Andy Sch\"urr\thanks{This work has been co-funded by the DFG within the Collaborative Research Center (CRC) 1053 -- MAKI.}
\institute{ 
Technische Universit\"at Darmstadt\\
Real-Time Systems Lab\\ 
Merckstr. 25\\
64283 Darmstadt, Germany\\
}
\and
\email{\{geza.kulcsar|frederik.deckwerth|malte.lochau|gergely.varro|andy.schuerr\}@es.tu-darmstadt.de}
} 
\def\titlerunning{Improved Conflict Detection for Graph Transformation with Attributes}
\def\authorrunning{G{\'e}za Kulcs{\'a}r, Frederik Deckwerth, Malte Lochau, Gergely Varr\'o, Andy Sch\"urr}

\begin{document}
\maketitle 
\begin{abstract}

In graph transformation, a conflict describes a situation where two alternative transformations cannot be arbitrarily serialized. When enriching graphs with attributes, existing conflict detection techniques typically report a conflict whenever at least one of two transformations manipulates a shared attribute. In this paper, we propose an improved, less conservative condition for static conflict detection of graph transformation with attributes 
by explicitly taking the semantics of the attribute operations into account. The proposed technique is based on symbolic graphs, which extend the traditional notion of graphs by logic formulas used for attribute handling. The approach is proven complete, i.e., any potential conflict is guaranteed to be detected.

\end{abstract}

\section{Introduction}
\label{01_Intro}

According to the \emph{Model-Driven Engineering} (MDE) principle, systems under design are represented by graph-based models. The change and evolution of such models is frequently described by the declarative, rule-based approach of \emph{graph transformation} \cite{Fundamentals,Handbook1}. 
However, models arising in real-world application scenarios typically contain numerical as well as textual attributes in addition to the graph-based structure.
For this purpose, an extension to graph transformation is required, being capable of representing and manipulating attributes of nodes and edges.


A major challenge in graph transformation is to statically analyse possible conflicts between rule applications.
The goal of conflict detection is to check if two 
graph transformation rules, both potentially applicable concurrently on the same input graph, are in any case arbitrarily serializable, 
i.e., if the two possible execution sequences result in the same (or at least two isomorphic) output graph(s). 

Critical Pair Analysis (CPA) is a common static analysis technique for conflict detection, defining 
a process of pairwise testing a set of graph transformation rules for possible conflicts~\cite{Fundamentals}. 
Unfortunately, a na\"ive adoption of CPA to graph transformation with attributes is too strict: 
whenever an attribute is modified by a rule application, and another rule application is 
also accessing the same attribute, they are immediately considered to be in conflict \cite{ConflOfAG}.

In this paper, we propose an improved, less conservative condition for static conflict detection of graph transformation with attributes 
by explicitly taking the semantics of the attribute operations into account. In particular, we make the following contributions: 
\begin{itemize}
\item{We define direct confluence as an appropriate conflict condition for graph transformation with attributes based on symbolic graphs, which reduces the number of false positives compared to existing conflict detection approaches. Using symbolic graphs further allows for an effective implementation of the proposed approach using a combination of graph transformation tools and off-the-shelf SMT solvers.}
\item{We prove 
that our approach is still complete \cite{Fundamentals}, i.e., any potential conflict is guaranteed to be detected.}
\end{itemize}

The paper is organized as follows: the basic concepts and definitions are introduced in Section 2. 
Section 3 proposes direct confluence as an improved conflict condition for rules with attributes and, based on that, \emph{\NCP pairs} are defined. 
In Section 4, the procedure for identifying conflicts is presented 
and proven complete. 
Section 5 surveys related work and Section 6 
concludes the paper.
\section{Preliminaries}
\label{02_Preliminaries}


In this section, we recapitulate the notions of symbolic graphs and symbolic graph transformation \cite{OL10a} that are used as a framework for our approach. 
Before getting into details of symbolic attributed graphs, we first define graphs and graph transformation without attributes.
\begin{defi}[Graphs and Graph Morphisms]\label{plaingr}
\normalfont
A \emph{graph} $G = (V_G,E_G,s_G,t_G)$ is a tuple consisting of a set of \emph{graph nodes} $V_G$, a set of \emph{graph edges} $E_G$, and 
the \emph{source} and \emph{target} functions $s_G,t_G:E_G \to V_G$, respectively.
A \emph{graph morphism} $f=(f_V,f_E):G \to H$, for mapping a graph $G$ to a graph $H$, consists of two functions 
$f_V:V_G \to V_H$ and $f_E:E_G \to E_H$ preserving the source and target functions: $f_V \circ s_G = s_H \circ f_E$ and 
$f_V \circ t_G = t_H \circ f_E$.
A graph morphism  is a \emph{monomorphism} if $f_V$ and $f_E$ are injective functions.
A graph morphism  is an \emph{isomorphism} if $f_V$ and $f_E$ are bijective functions.
\end{defi}

Based on this definition of graphs, graph transformation relies on the notion of \emph{pushouts}. 
A pushout has the following 
meaning (in the category of graphs): given three graphs $A,B,C$ and two morphisms $f:A \to B, g:A \to C$, their pushout consists of the \emph{pushout object} $P$ 
and two morphisms $g':B \to P, f':C \to P$, where $P$ is the \emph{gluing} of $B$ and $C$ along the elements of $A$, the latter being, 
in a way, present in both as $f(A)$ and $g(A)$, respectively. 
Correspondingly, {pullbacks} are the counterpart of pushouts. Given three graphs $B,C,P$ and two morphisms $g':B \to P, f':C \to P$, their pullback consists of the \emph{pullback object} $A$ and the morphisms $f:A \to B, g:A \to C$, 
where $A$ can be seen as the intersection of $B$ and $C$, i.e., the elements of $B$ and $C$ which are overlapping in $P$.

In the following, we use the double pushout (DPO) approach to define graph transformation \cite{Fundamentals,Handbook1}.

\begin{defi}[Graph Transformation Rule]
\normalfont

A \emph{graph transformation rule} $r$ in the DPO approach consists of a left-hand side (LHS) graph $L$, an interface graph $K$, and a right-hand side (RHS) graph $R$ and the morphisms $l:K \to L$ and $r:K \to R$. 
\end{defi}

An \emph{application} of rule $r$ to a graph $G$ is defined by the two pushouts $(1)$ and $(2)$ in the diagram below:

\begin{figure}[H]
\centering
\includegraphics[width=0.28\textwidth]{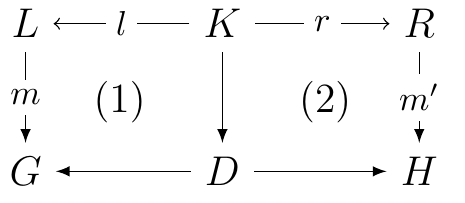}
\end{figure}

A rule is applied by first identifying a \emph{match} $m:L\to G$ of the left-hand side $L$ in graph $G$. 
In the next step, the \emph{context graph} $D$ is obtained by removing all elements in $G$ which are identified by match $m$,
but are not contained in the interface $K$. The result of the rule application, $H$, is obtained by 
adding all elements of the right-hand side $R$ to the context $D$ which do not have a pre-image in the interface $K$. 

A \emph{direct derivation} of rule $r$ at match $m$, denoted as $G \overset{r,m}{\Longrightarrow} H$, is the single step from 
graph $G$ to graph $H$ derived by applying rule $r$ to graph $G$ at the match $m$.

Until now, we have limited our discussion to plain graphs, i.e., graphs incapable of expressing attributes such as integer variables with corresponding operations.
As a first step towards
graphs with attributes, we extend their definition to E-graphs \cite{Fundamentals}. 
An E-graph is a graph extended by special kinds of label nodes ($V^D$) and edges ($E^{VL}$ and $E^{EL}$ for node and edge attribution, respectively) used for carrying the attribute values.

\begin{defi}[E-graphs and E-graph Morphisms \cite{Fundamentals}]
\normalfont
An 
\emph{E-graph} $EG = (G,D)$ is a tuple consisting of a graph $G$ and a \emph{labeling part} $D = (V^D_G,E^{VL}_G,E^{EL}_G,s^{VL}_G,
t^{VL}_G,s^{EL}_G,t^{EL}_G)$ with a set of \emph{label nodes} $V^D_G$, two sets of edges $E^{VL}_G$ and $E^{EL}_G$ for node and edge labeling, respectively, 
and the source and target functions $s^{VL}_G:E^{VL}_G \to V_G$, $t^{VL}_G:E^{VL}_G \to V^D_G$, $s^{EL}_G:E^{EL}_G \to E_G$ and $t^{EL}_G:E^{EL}_G \to V^D_G$ assigning the label nodes to the graph nodes an edges, respectively.

An \emph{E-graph morphism} $h = (h_G, h_D, h_{VL}, h_{EL})$ consists of a graph morphism $h_G$ and three functions 
$h_D, h_{VL}, h_{EL}$
mapping the label nodes and the labeling edges while preserving source and target functions.
An E-graph morphism is a \emph{monomorphism} (\emph{isomorphism}) if its functions are injective (bijective).
\end{defi}

In the following, we omit the $E$- prefix and denote E-graphs using e.g. $G$ instead of $EG$.

The construction of E-graphs contains labels as placeholders for attribute values. In order to be able to 
define and manipulate those attribute values, we employ a data algebra. 
A \emph{data algebra} $\mathcal{D}$ is a signature $\Sigma$ consisting of symbols for sorts, functions and predicates; and a mapping of these symbols to sets and functions, assigning meaning to the symbols.
For the examples, we use the algebra of natural numbers with addition and equality. This algebra consists of the sort symbol $\mathbb{N}$ representing the (infinite) set of natural numbers,  the binary function symbol '$+$' mapped to addition with the usual meaning, and the binary predicate symbol '$=$' defined by the equality relation on $\mathbb{N}$. 
For further details we refer to \cite{EhrigMahr85}.
%
%

%



The concept of symbolic graphs has been introduced recently to combine the concept of E-graphs for representing attributes and data algebras for the values of those attributes. This way, symbolic graphs provide 
a convenient representation of 
graphs with attributes \cite{OL10a}. 
In particular, a symbolic graph is an E-graph whose label nodes contain variables and the values of these variables are constrained by a first-order logic formula, also being part of the symbolic graph.

Given a $\Sigma$-algebra $\mathcal{D}$ and a set of variables $\mathcal{X}$,
a \emph{first-order logic formula} is built from the variables in $\mathcal{X}$, the
function and predicate symbols in $\Sigma$, the logic operators $\vee,\wedge,\neg,\Rightarrow,\Leftrightarrow$, 
the constants \textit{true} and \textit{false} 
and the quantifiers $\forall$ and $\exists$ in the usual way \cite{Sf67}. A \emph{variable
assignment} $\sigma:\mathcal{X} \to \mathcal{D}$ maps the variables $x \in \mathcal{X}$ to a value in
$\mathcal{D}$. A first-order logic formula $\Phi$ is \emph{evaluated} for a given assignment $\sigma$ 
by first replacing all variables in $\Phi$ according to the assignment
$\sigma$ and evaluating the functions and predicates according to the algebra, and the
logic operators. 
We write $\mathcal{D},\sigma \models \Phi$ if and only if
$\Phi$ evaluates to \textit{true} for the assignment $\sigma$; and $\mathcal{D}
\models \Phi$, if and only if $\Phi$ evaluates to \textit{true} for all assignments.


\begin{defi}[Symbolic Graphs and Symbolic Graph Morphisms \cite{OL10a}]
\normalfont
A \emph{symbolic graph} $SG = (G, \Phi_G)$ consists of an E-graph $G$ and
a first-order logic formula $\Phi_G$ over a given data algebra $\mathcal{D}$, using the label nodes of $G$ 
as variables and elements of $\mathcal{D}$ as constants. 

A \emph{symbolic graph morphism} $h: (G, \Phi_G) \to (H, \Phi_H)$ 
is an E-graph morphism $h:G \to H$ such that $\mathcal{D} \models \Phi_H \Rightarrow h_\Phi(\Phi_G)$, where $h_\Phi(\Phi_G)$ 
is the first-order logic formula obtained when replacing each
variable $x$ in formula $\Phi_G$ as defined by the mapping for the label nodes $h_D(x)$.
The symbolic graphs $SG_1 = (G_1,\Phi_1)$ and $SG_2 = (G_2,\Phi_2)$ are isomorphic if there is a 
symbolic graph morphism $h: SG_1 \to SG_2$ that is an E-graph isomorphism and $\mathcal{D} \models h_\Phi(\Phi_1) \Leftrightarrow \Phi_2$.
\end{defi}
As the variables and, thus, the attribute values are determined by a first-order logic formula, a symbolic graph can be seen as a class of \emph{grounded symbolic graphs} (GSG). A \emph{grounded symbolic graph} is a symbolic graph where\begin{inparaenum}[(i)]
\item each attribute value is constant, and
\item for each value of the data algebra, it contains a corresponding constant label node\end{inparaenum}. A grounded symbolic graph is created by adding to the set of label nodes a variable $c_v$ for each value $v$ in $\mathcal{D}$, and extending the formula with the equation $c_v=v$, which assigns a constant value to each constant variable.

\begin{defi}[Grounded Symbolic Graph \cite{OL10a}]
\normalfont
A symbolic graph $SG=(G,\Phi_{G})$ with data algebra $\mathcal{D}$ is grounded, denoted as $\underline{SG}$, if it includes a variable $c_v \in V^D_{G}$ for each value $v\in \mathcal{D}$, and for each variable assignment $\sigma : V^{D}_{G} \to \mathcal{D}$ such that $\mathcal{D}, \sigma \models \Phi_{G}$, it holds that $\sigma(c_v)=v$.
\end{defi}
A grounded symbolic graph $\underline{SH}$ is an \emph{instance} of a symbolic graph $SG$ via $h: SG \to \underline{SH}$ if $h$ is a symbolic graph morphism, which is injective for all kinds of nodes and edges except the label nodes.
\begin{figure} [tp]
\centering
\subfloat[Grounded Symbolic Graph \label{fig: ex0_symbgraph}]{%
\includegraphics[height=0.14\textwidth]{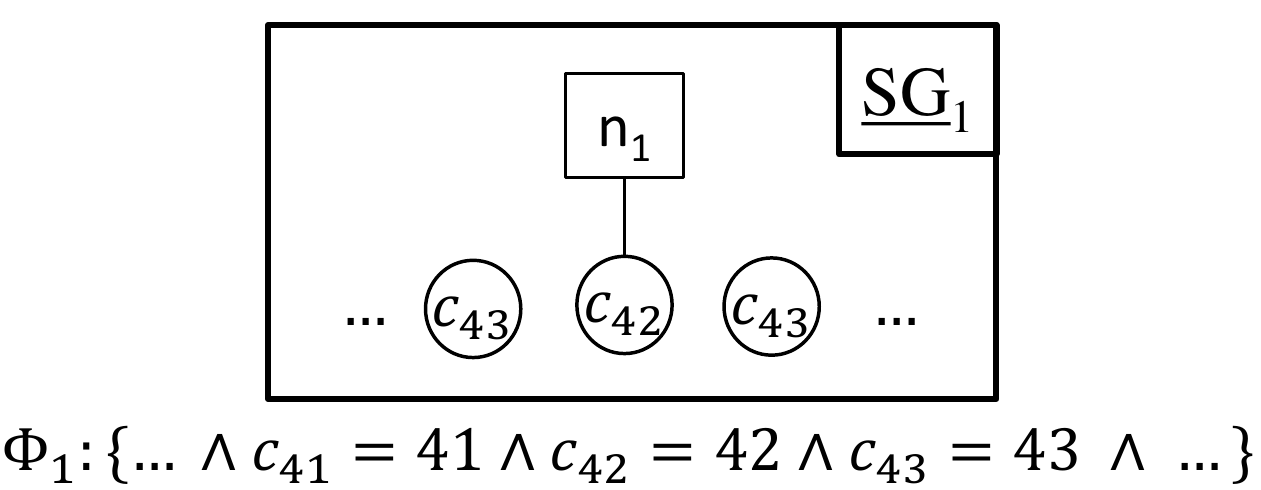}}\quad
\subfloat[Symbolic Graph \label{fig: ex01_symbgraph}]{%
\includegraphics[height=0.14\textwidth]{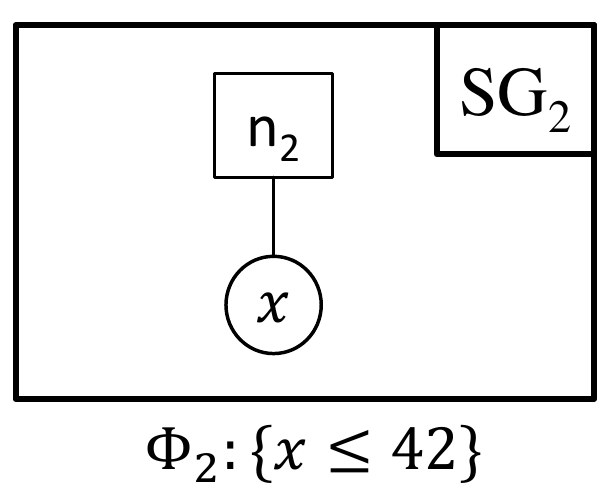}}\quad
\caption{Example of a Grounded Symbolic Graph and a (Non-Grounded) Symbolic Graph}
\end{figure}
\begin{exmp}[Symbolic and Grounded Symbolic Graphs]\normalfont
Figure \ref{fig: ex0_symbgraph} shows a grounded symbolic graph $\underline{SG}_1=(G_1,\Phi_1)$ consisting of a single graph node $\mathsf{n}_1$ bearing an attribute carrying the variable $c_{42}$, and a formula $\Phi_1$, constraining each variable $c_{v}$ to value $v\in D$. The grounded symbolic graph $\underline{SG}_1$ contains an infinite number of label nodes and corresponding equations as indicated by the '$\ldots$' in Figure~\ref{fig: ex0_symbgraph}. 

Figure \ref{fig: ex01_symbgraph} shows the (non-grounded) symbolic graph $SG_2=(G_2,\Phi_2)$ whose E-graph part is identical to $G_1$. Consequently, there exists an E-graph morphism $h:G_2\to G_1$ mapping nodes $\mathsf{n}_2$ and $x$ of $G_2$ to nodes $\mathsf{n}_1$ and $c_{42}$ of $G_1$, respectively. This morphism is a valid symbolic graph morphism as, according to the mapping of the label nodes ($h_\Phi(c_{42})=x$), the condition $\Phi_1 \Rightarrow h_\Phi(\Phi_2)$ can be simplified to $(x=42) \Rightarrow (x\le 42)$ which evaluates to true. Hence, the grounded symbolic graph $\underline{SG}_1$ is an instance of the symbolic graph $SG_2$.   
\end{exmp}

Pushouts and pullbacks in symbolic graphs can be defined in terms of pushouts and pullbacks for graphs \cite{OL10a}. More specifically, the symbolic morphisms $f:(A,\Phi_A) \to (B,\Phi_B)$ and $g:(A,\Phi_A) \to (C,\Phi_C)$ are a symbolic pushout $f':(B,\Phi_B) \to (D,\Phi_D)$ and $g':(C,\Phi_C) \to (D,\Phi_D)$ with pushout object $(P,\Phi_P)$ if $f'$ and $g'$  are a pushout in E-graphs and $\mathcal{D}\models(\Phi_P\Leftrightarrow f'_\Phi(\Phi_A) \land g'_\Phi(\Phi_C)$). A pullback is defined analogously where the formula $\Phi_A$ of the pullback object is given by the disjunction of $\Phi_B$ and $\Phi_C$.  

A symbolic graph transformation rule is a graph transformation rule additionally equipped with a first-order logic formula.
\begin{defi}[Symbolic Graph Transformation Rule and Symbolic Direct Derivation \cite{OL10a}]
\normalfont

A \emph{symbolic graph transformation rule} $r$ is a pair $(L \overset{l}{\leftarrow} K \overset{r}{\rightarrow} R, \Phi)$, where 
$(L \overset{l}{\leftarrow} K \overset{r}{\rightarrow} R)$ is an E-graph transformation rule and $\Phi$ is a single first-order logic formula shared by $L$, $K$ and $R$. The E-graph morphisms $l$ and $r$ are of a class $\mathcal{M}$ of morphisms injective for graph nodes and all kinds of edges and bijective for label nodes. 

A \emph{symbolic direct derivation} $SG \overset{r,m}\Longrightarrow SH$ is the application of a symbolic rule 
$r = (L \overset{l}{\leftarrow} K \overset{r}{\rightarrow} R, \Phi)$ on the symbolic graph $SG = (G,\Phi_G)$ 
at match $m: L \to SG$, resulting in the symbolic graph $SH = (H,\Phi_H)$, where $m$ is a symbolic graph morphism, which is injective for all kinds of nodes and edges except for the label nodes and $SH$ is 
produced as a DPO diagram in E-graphs.
\end{defi}

\begin{fact}[Properties of Symbolic Direct Derivations \cite{Lazy}]
\label{fac: Properties of symbolic direct derivations}
\normalfont
The restrictions on morphisms $l$ and $r$ ensure that for any symbolic direct derivation $SG \overset{r,m}\Longrightarrow SH$,
\begin{enumerate}[(i)]
  \item the set of label nodes and the formula remain unaltered, i.e., $V^D_G =V^D_H$ and $\mathcal{D}\models \Phi_G \Leftrightarrow \Phi_H$, and
  \item if $SG$ is grounded, then so is $SH$.
\end{enumerate}  
\end{fact}

\begin{figure}[tp]
\centering
\includegraphics[width=0.45\textwidth]{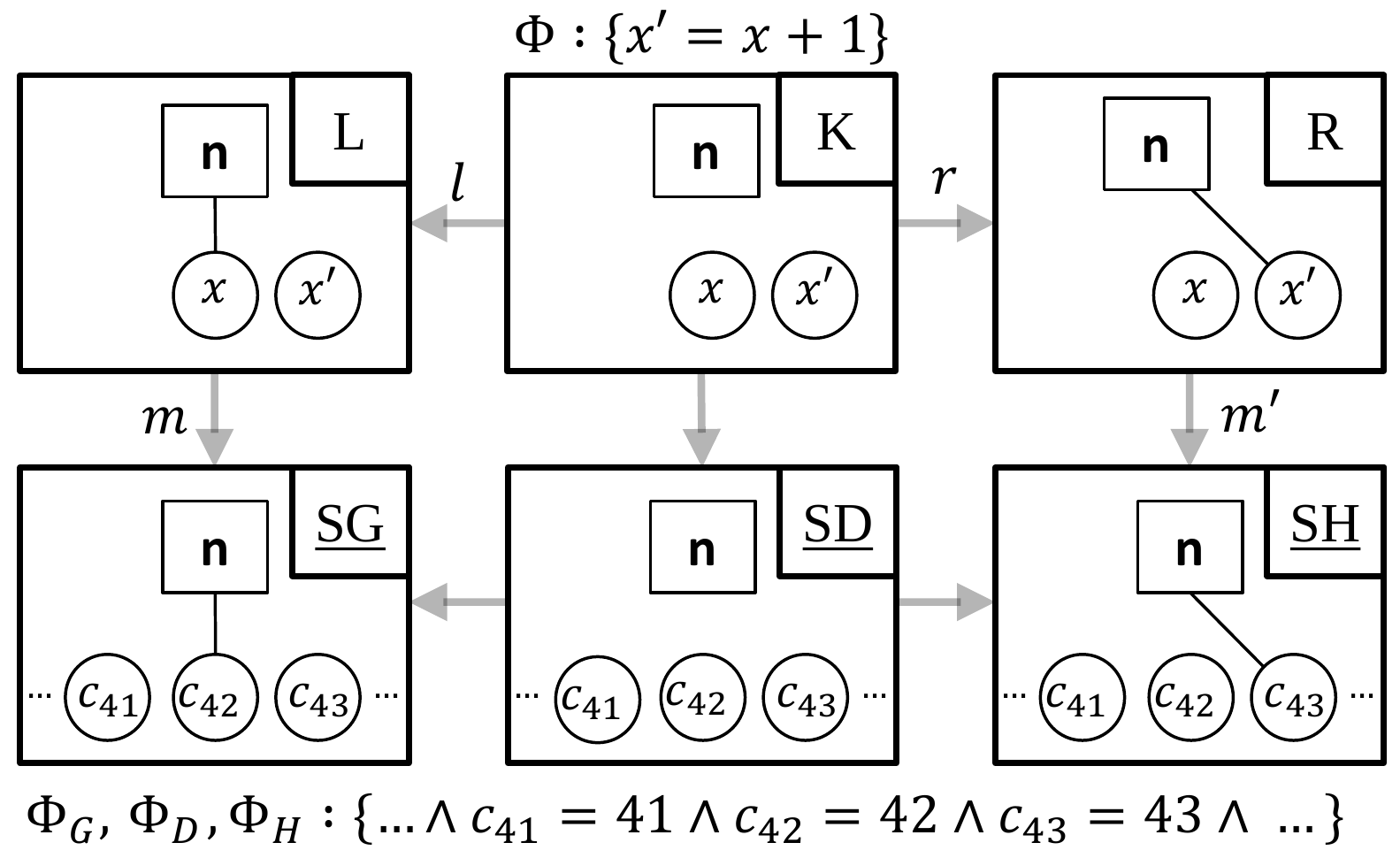}
\caption{Example of a Symbolic Direct Derivation}
\label{fig: Example of a symbolic direct derivation}
\end{figure}
Note that (i) also implies coincidence on label nodes of the match $m: L \to SG$ and the comatch $cm: R\to SH$, i.e., $m_\Phi=cm_\Phi$. 

Although, it seems counterintuitive at a first glance that we require $L$, $K$
and $R$ to share the same formula and set of label nodes, it does not mean that
attribute values cannot be changed by a rule application, since attribute values
are modified by redirecting label edges.
\begin{exmp}[Symbolic Graph Transformation Rule and Symbolic Direct Derivation]
\normalfont
Figure \ref{fig: Example of a symbolic direct derivation} shows a symbolic graph transformation rule
$r = (L \overset{l}{\leftarrow} K \overset{r}{\rightarrow} R, \Phi)$ (depicted in the upper part). The rule takes a graph node \textsf{n} that has at least one attribute (denoted by the label edge between \textsf{n} label node $x$) and increases it by one. This is achieved by introducing a new label node $x'$ to represent the attribute value after the rule application and constraining it to $x'=x+1$ as defined by the formula $\Phi$. The attribute value is changed from the old value $x$ to the new value $x'$ by first deleting the label edge between \textsf{n} and the old value $x$ and afterwards creating a new label edge assigning the new value $x'$ to \textsf{n}.
The result from applying the rule to grounded symbolic graph $\ul{SG}$ is shown on the bottom of Figure~\ref{fig: Example of a symbolic direct derivation}. The only valid mapping for match $m$ to satisfy $\Phi_G\Rightarrow m_\Phi(\Phi)$ is to map $x$ to $c_{42}$ and $x'$ to $c_{43}$. Then the resulting direct derivation $\ul{SG} \overset{r,m}\Longrightarrow \ul{SH}$ changes the attribute value from $42$ (in grounded symbolic graph $\ul{SG}$) to $43$ in grounded symbolic graph $\ul{SH}$ as expected.   
\end{exmp}

%
%
%

In the following, we use symbolic graphs and symbolic graph transformation to present our approach.

\section{A Conflict Notion for Graph Transformation with Attributes}
\label{03_Concept}

In this section, we present an improved detection technique for potential \emph{rule conflicts} for graph transformation with attributes.  
To this end, we define a notion of \emph{conflict on the level of direct derivations}, and we review parallel dependence as an existing sufficient condition for our notion of conflict. Thereupon, we show by means of an illustrative example that parallel dependence is too conservative especially in an attributed setting, i.e., rejecting too many conflict-free direct derivations. 
To overcome these deficiencies, we present a new condition, called \emph{direct confluence}, 
that is sufficient for detecting conflicting direct derivations, but less restrictive than parallel dependence. 
Finally, to reason about conflicts on the rule level, we lift the direct confluence condition by defining conflicting pairs.

%
%

%
With the concept of \emph{conflicts}, we grasp the situation where, given two rules ($r_1$ and $r_2$) applicable on the same graph, we obtain different results depending on which rule is applied first.
We characterize a conflict in terms of two alternative direct derivations that can not be arbitrarily serialized. In this case, applying the 
second transformation after the first leads to a different result than vice versa.
\begin{defi} [Conflict]\label{def: Conflict}
\normalfont
Given a grounded symbolic graph $\underline{SG}$, the two alternative direct derivations $\underline{SH}_1 \overset{r_1,m_1}{\Longleftarrow}\underline{SG}\overset{r_2,m_2}{\Longrightarrow}\underline{SH}_2$ are a \emph{conflict} if no direct derivations $\underline{SH}_1\overset{r_2,m_2'}{\Longrightarrow}\underline{SX}_1$ and  $\underline{SH}_2\overset{r_1,m_1'}{\Longrightarrow}\underline{SX}_2$ exist with $\underline{SX}_1$ and $\underline{SX}_2$ being isomorphic. 
\end{defi}
\noindent Note that since $\underline{SG}$ is grounded, $\underline{SH}_1$, $\underline{SH}_2$, $\underline{SX}_1$ and $\underline{SX}_2$ are grounded, too.

This definition of conflicts leaves open how to practically determine that two given alternative direct derivations are a conflict.
A corresponding condition to check if two direct derivations are a conflict is referred to as a \emph{conflict condition}.

\subsection{Parallel Dependence as a Conflict Condition}
\label{sec: Parallel Dependence}
In the literature of graph transformation, a common conflict condition is the notion of \emph{parallel dependence} \cite{Handbook1, Fundamentals}. 
Intuitively, two direct derivations are parallel dependent if they are mutually exclusive, i.e., after one of the direct derivations, the other rule is not applicable anymore and/or vice versa. We adapt the notion of parallel dependence to symbolic graphs as follows.

\begin{defi}[Parallel Dependence]\label{dependence}
\normalfont
The symbolic direct derivations $(H_1,\Phi) \overset{r_1,m_1}{\Longleftarrow} (G,\Phi)\overset{r_2,m_2}{\Longrightarrow} (H_2,\Phi)$ 
are parallel dependent iff the direct (E-graph) derivations $H_1 \overset{r_1,m_1}{\Longleftarrow} G \overset{r_2,m_2} {\Longrightarrow}H_2$ are parallel dependent, i.e., there does not exist E-graph morphism $i:L_1\rightarrow D_2$ or $j:L_2\rightarrow D_1$ such that $m_1 = g_2 \circ i$ and $m_2 = g_1 \circ j$, as in the diagram below.
\begin{figure}[H]
\begin{center}
  \includegraphics[width=0.58\textwidth]{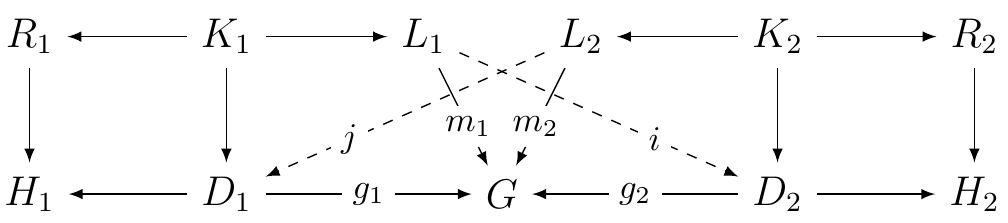}
\end{center}
\end{figure}
%
\noindent Two direct derivations not being parallel dependent are called \emph{parallel independent}.
\end{defi}
\noindent Note that the non-existence of morphism $i$ means that the application of rule $r_2$ deletes at least one element which is required for the match of $r_1$ and vice versa for $j$.  
  
\begin{figure}
\centering
\includegraphics[width=0.95\textwidth]{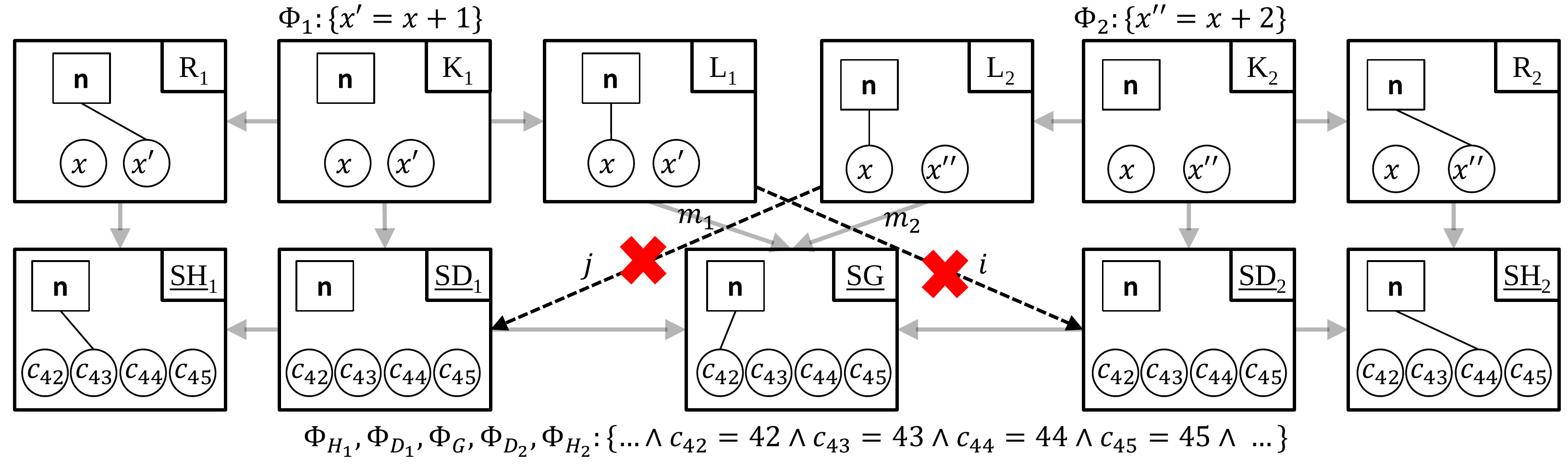}
\caption{Example of Parallel Dependent Direct Derivations}
\label{figpardep}
\end{figure}

\begin{exmp}[Parallel Dependence]\label{expardep}
\normalfont

Figure \ref{figpardep} shows an example of two parallel dependent direct derivations. 
The two symbolic rules $r_1 = (L_1 \overset{l_1}{\leftarrow} K_1 \overset{r_1}{\rightarrow} R_1,\Phi_1)$ and $r_2 = (L_2 \overset{l_2}{\leftarrow} K_2 \overset{r_2}{\rightarrow} R_2,\Phi_2)$ are shown in the upper part of the figure. Both rules take a single graph node \textsf{n} with a single attribute (label node $x$); while rule $r_1$ increases the value of the attribute by $1$, rule $r_2$ adds $2$ to the attribute value. The bottom part of Figure~\ref{figpardep} shows the application of the rules on the grounded symbolic graph $\ul{SG}$.
As the morphisms $i:L_1\rightarrow D_2$ and $j:L_2\rightarrow D_1$ do not exist because of a missing labeling edge, the depicted direct 
derivations are parallel dependent and, therefore, they are declared to be a conflict by parallel dependence. 

However, if focusing on the intention of these rules, 
it seems rather intuitive that the direct derivations are not a conflict as the operations expressed by the rules are commutative, i.e., $x+1+2=x+2+1$.
\end{exmp}

Concluding our example, although this technique is practical, efficient and only the two direct derivations are required for the decision process, it seems too strict (i.e., it produces too many false positives) for the desired attributed setting. The problem is that using the notion of parallel dependence, two rules are considered to have a (potential) conflict whenever an attribute is modified by one rule, that is accessed by the other rule (as also stated in \cite{ConflOfAG}). The root of the problem resides in the construction of the underlying E-graphs, which do not reflect the intention of attribute operations, but rather delete and recreate the labeling edges whenever a new value is assigned to an attribute.

\subsection{Direct Confluence as an Improved Conflict Condition}
\label{sec: Direct Confluence}

To overcome the deficiencies of parallel dependence as a conflict condition, we propose an alternative approach. 
Our proposal is based on the observation that the definition of conflicts (Def. \ref{def: Conflict}) allows for \emph{directly checking if the different application 
sequences of the two rules result in isomorphic graphs}. In particular, the proposed approach relies on 
our notion of \emph{direct confluence}. To be more precise, 
two direct derivations which are \emph{not directly confluent} are a conflict.

The definition of direct confluence has to fulfill that (i) given a pair of direct derivations for two rules $r_1$ and $r_2$ on the same input graph, there exists two derivation sequences (i.e. first $r_1$ and then $r_2$ and vice versa) whose resulting graphs are isomorphic and (ii) in both derivation sequences, the second direct derivations preserves at least the elements as the first direct derivations and send these to the same elements in the common result. 

\begin{defi}[Direct Confluence]\label{def: dirconfl}
\normalfont

Given a pair of direct derivations $SH_1 \overset{r_1,m_1}{\Longleftarrow} SG \overset{r_2,m_2}{\Longrightarrow} SH_2$ with $SG = (G,\Phi_G)$, 
$SH_1 = (H_1,\Phi_{H_1})$ and $SH_2 = (H_2,\Phi_{H_2})$ being symbolic graphs, 
they are \emph{directly confluent} 
if there exist direct derivations $SH_1 \overset{r_2,m'_2}{\Longrightarrow} SX_1$ and 
$SH_2 \overset{r_1,m'_1}{\Longrightarrow} SX_2$ such that 
\begin{enumerate}[I.]
\item\label{cond1} $SX_1 = (X_1,\Phi_{X_1})$ and $SX_2 = (X_2,\Phi_{X_2})$ are isomorphic, and
\item\label{cond2} matches $m'_1$ and $m'_2$ are chosen in a way that (2), (3) and (4) commute, where 
(1) is the pullback of ($SD_1 \rightarrow SG \leftarrow SD_2$) and the graphs $SD_1$, $SD_2$, $SQ_1$ and $SQ_2$ are the context graphs of the corresponding direct derivations.
\end{enumerate}
\begin{figure} [H]
\centering
    \subfloat[Property \ref{cond1}]{\label{fig:sf1}
      \includegraphics[height=0.24\textwidth]{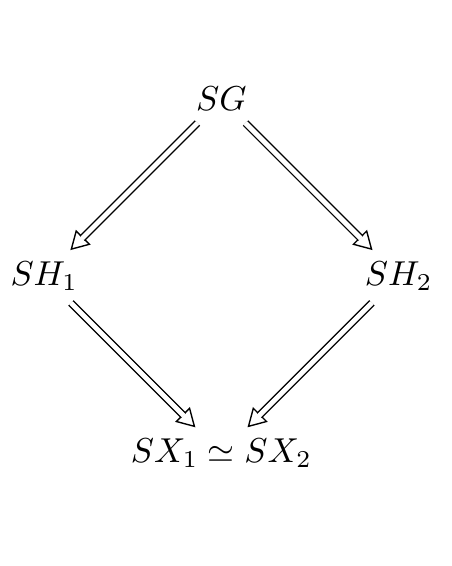}} 
    \subfloat[Property \ref{cond2}]{\label{fig:sf2}
      \includegraphics[height=0.24\textwidth]{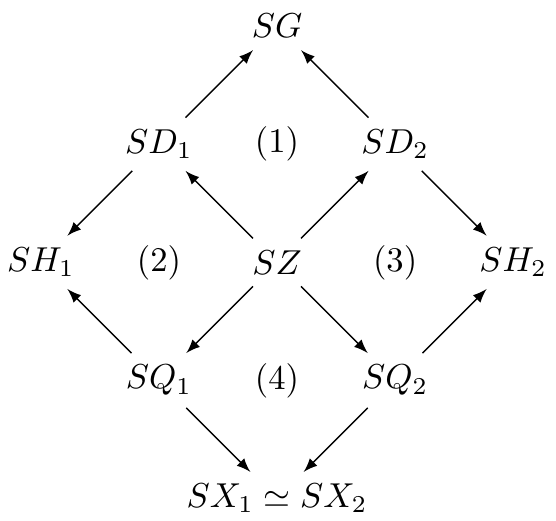}} 
\end{figure}
\end{defi}

Property \ref{cond1} ensures that the given direct derivations are not a conflict. 
Property \ref{cond2} serves as a means of tracking for the matched elements after the direct derivations. This way, it is guaranteed that the second direct derivations are applied to the images of the same elements as the first ones. In other words, the symbolic graph $SZ$ contains all elements from the input graph that are preserved by \emph{both} original direct derivations and the commuting rectangles of Property~\ref{cond2} guarantee that these elements are in the context graphs of the second direct derivations and (through the lower rectangle) that they are embedded in the resulting graph in the same way. In the following, when using 
the concept of direct confluence, we always assume that the matches are chosen appropriately according to Property \ref{cond2}. Note that the definition of direct confluence is a specialization of strict confluence as defined in~\cite{Fundamentals} (Def.~6.26), with the lower transformation chains consisting of exactly one direct derivation.
\begin{figure}[htbp]
\centering
\includegraphics[width=0.95\textwidth]{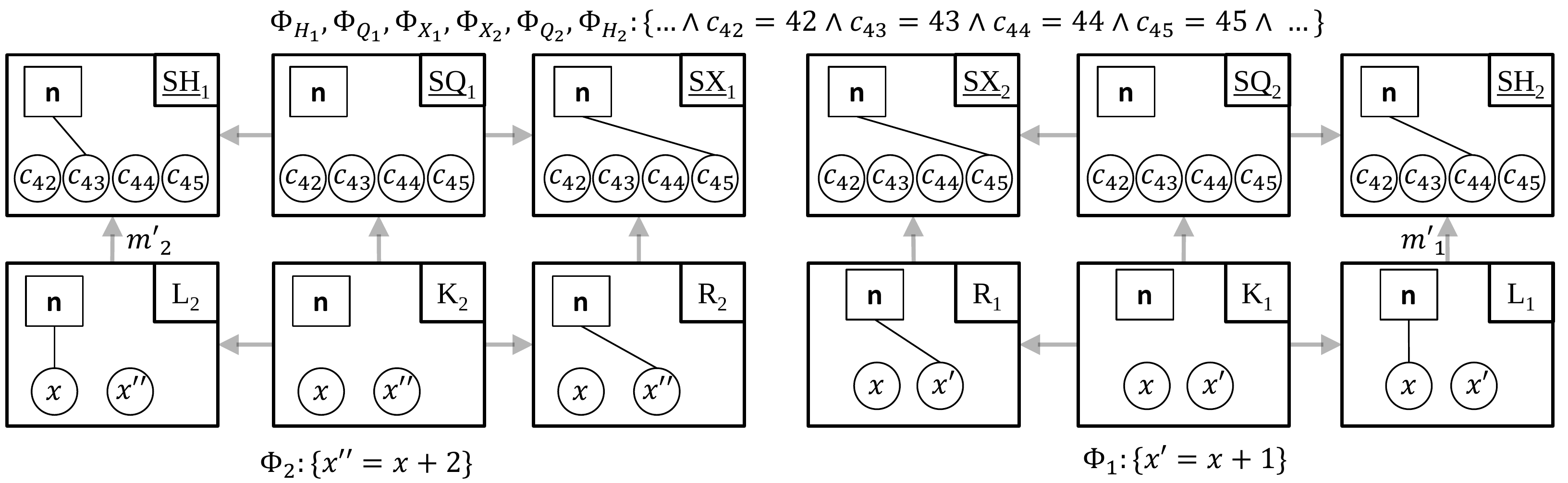}
\caption{Example of Direct Confluence}
\label{fig: Example of direct confluence}
\end{figure}

\begin{exmp}[Direct Confluence as an Improved Conflict Condition]\label{ex: Direct confluence as a refined conflict condition}
\normalfont

Figure~\ref{fig: Example of direct confluence} shows (in the top right and top left corner) the results $\ul{SH}_1$ and $\ul{SH}_2$ of the alternative direct derivations $\ul{SH}_1 \overset{r_1,m_1}{\Longleftarrow} \ul{SG} \overset{r_2,m_2} {\Longrightarrow}\ul{SH}_2$ presented in Example~\ref{expardep} (shown in Figure~\ref{figpardep}). On the bottom (from left to right), the symbolic rules $r_2 = (L_2 \overset{l_2}{\leftarrow} K_2 \overset{r_2}{\rightarrow} R_2,\Phi_2)$  and $r_1 = (L_1 \overset{l_1}{\leftarrow} K_1 \overset{r_1}{\rightarrow} R_1,\Phi_1)$ are shown. In order to check direct confluence, both rules are applied to $\ul{SH}_1$ and $\ul{SH}_2$, resulting in the direct derivations  $\ul{SH}_1 \overset{r_2,m'_2}{\Longrightarrow} \ul{SX}_1$ and $\ul{SH}_2 \overset{r_1,m'_1}{\Longrightarrow} \ul{SX}_2$. As grounded symbolic graphs $\ul{SX}_1$ and $\ul{SX}_2$ are isomorphic, direct confluence declares, in contrast to parallel dependence, that the two alternative derivation $\underline{SH}_1 \overset{r_1,m_1}{\Longleftarrow}\underline{SG}\overset{r_2,m_2}{\Longrightarrow}\underline{SH}_2$ are not a conflict.

\end{exmp}

We have shown that direct confluence as a conflict condition is in accordance with our notion of conflicts and is, therefore, suitable for conflict detection in the presence of attributes. However, in most applications, one is rather interested in a conflict detection on the level of rules instead of their applications.

\subsection{Lifting Conflicts to Rule Level}

In the following, we show how we lift our notion of direct confluence from the direct derivation level to the rule level. As a starting point, we recall the well-known concept of \emph{critical pairs} that is used to lift the parallel dependence condition to the rule level. First, we adapt critical pairs to our setting of symbolic graphs. Afterwards, we show that this criterion is too conservative, however, it is used as a first necessary condition in the decision process as if two rules are parallel independent, they are also directly confluent (note that this does not necessarily hold the other way around). To improve conflict detection, we proceed by showing how direct confluence can be lifted to an adequate rule conflict condition in the presence of attributes. 

A critical pair for two given rules 
consists of a \emph{minimal context} and two parallel dependent direct derivations.
%
A \emph{minimal context} of two rules is a graph (i) on which both rules are applicable and (ii) which only contains elements being matched by at least one of the rules. 
The intention behind critical pairs 
essentially consists in identifying those minimal conflict instances representing \emph{each} possible conflict of the rules on any possible input graph. Practically, this requirement means that whenever two direct derivations are a conflict on some graph $SG$, there is an element in the corresponding set of minimal conflict instances which is embedded in $SG$. \emph{Embedding} one pair of direct derivations (with input graph $SK$) into another pair of direct derivations (with input graph $SG$) means that there exist monomorphisms from the graphs of the first pair of derivations to the graphs of the second one.

The definition of critical pairs has only been considered in the framework of plain and attributed graphs before \cite{Fundamentals}. 
Nevertheless, it can be extended to symbolic graphs as follows.

\begin{defi}[Symbolic Critical Pair]\label{def: scp}
\normalfont

A pair of symbolic rule applications $SP_1 \overset{r_1,o_1}{\Longleftarrow} SK \overset{r_2,o_2}{\Longrightarrow} SP_2$ 
with rules $r_1 = (L_1 \overset{l_1}{\leftarrow} K_1 \overset{r_1}{\rightarrow} R_1, \Phi_1)$ and 
$r_2 = (L_2 \overset{l_2}{\leftarrow} K_2 \overset{r_2}{\rightarrow} R_2, \Phi_2)$ on the input graph $SK = (K,\Phi_K)$ 
is a \emph{symbolic critical pair} if it is parallel dependent, $\mathcal{D}\models\Phi_K \Leftrightarrow o_{1,\Phi}(\Phi_1) \wedge o_{2,\Phi}(\Phi_2)$, and $K$ is 
minimal meaning that each E-graph element $ge \in K$ (i.e., node or edge in $K$) has a pre-image in the LHS of rule $r_1$ or $r_2$, i.e., $ge \in o_1(L_1)$ or $ge \in o_2(L_2)$. 
\end{defi}

%

\begin{exmp}[Symbolic Critical Pair]\label{ex: Critical pair}
\normalfont
Figure~\ref{fig: Example of a critical pair} provides an example for a symbolic critical pair according to Definition~\ref{def: scp}. Again, we consider the rules $r_1$ and $r_2$ shown in the upper part of the figure. Contrary to the example for parallel dependence, the rules are now applied to the minimal context $SK$ that contains only the elements required for applying the rules $r_1$ and $r_2$. As the resulting pair of direct derivations $SP_1 \overset{r_1,o_1}{\Longleftarrow} SK \overset{r_2,o_2}{\Longrightarrow} SP_2$ can be embedded into the direct derivations $\ul{SH}_1 \overset{r_1,m_1}{\Longleftarrow} \ul{SG} \overset{r_2,m_2} {\Longrightarrow} \ul{SH}_2$ of Example~\ref{expardep}, the pair $SP_1 \overset{r_1,o_1}{\Longleftarrow} SK \overset{r_2,o_2}{\Longrightarrow} SP_2$ is a minimal conflict instance of the conflict $\ul{SH}_1 \overset{r_1,m_1}{\Longleftarrow} \ul{SG} \overset{r_2,m_2} {\Longrightarrow} \ul{SH}_2$.

\begin{figure}
\centering
\includegraphics[width=0.93\textwidth]{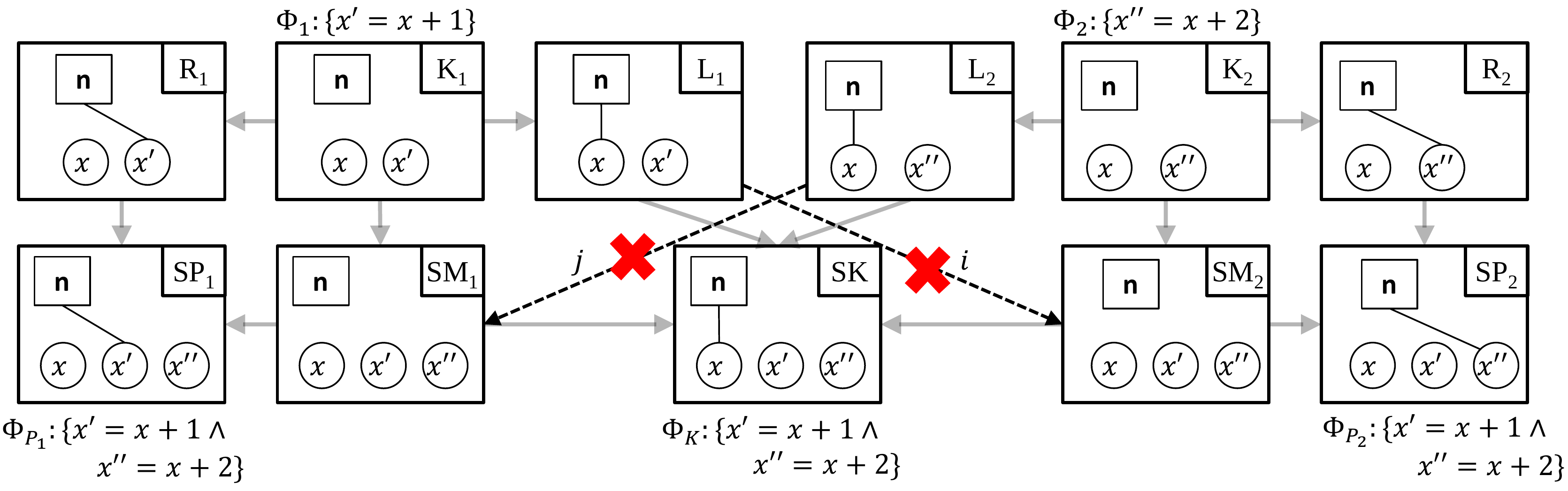}
\caption{Example of a Critical Pair}
\label{fig: Example of a critical pair}
\end{figure}

\end{exmp}

This example has shown that the parallel dependence condition can be lifted to rule level by the concept of symbolic critical pairs. Analogously, we also lift the direct confluence condition 
to the level of rules instead of direct derivations, using a construction similar to minimal contexts. Unfortunately, 
when considering (general) symbolic graphs and symbolic graph transformation, a general problem arises when checking direct confluence, as is illustrated in the following example.

\begin{exmp}[Problem of Checking Direct Confluence]\label{exdc}
\normalfont

\begin{figure}[htbp]
\centering
\includegraphics[width=0.93\textwidth]{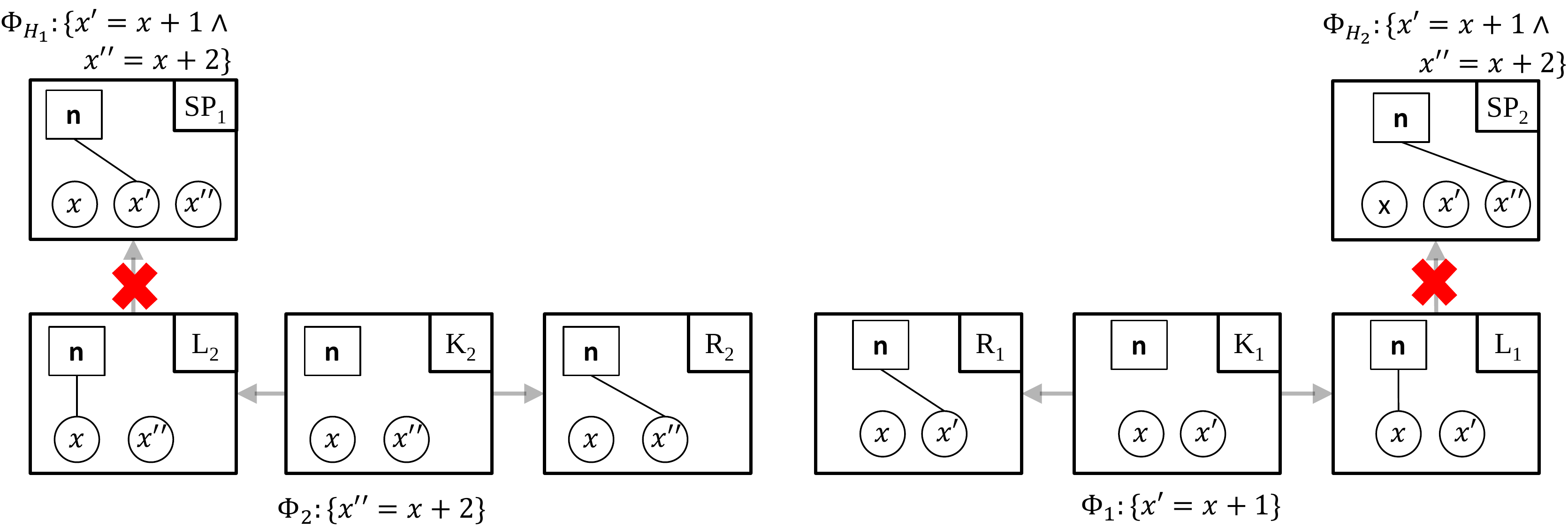}
\caption{Problem of Checking Direct Confluence}
\label{fig: Example of direct confluence2}
\end{figure}

Figure~\ref{fig: Example of direct confluence2} shows (in the upper part) the results $SP_1$ and $SP_2$ of the alternative direct derivations $SP_1 \overset{r_1,o_1}{\Longleftarrow} SK \overset{r_2,o_2} {\Longrightarrow}SP_2$ presented in Example~\ref{ex: Critical pair} (shown in Figure~\ref{fig: Example of a critical pair}). On the bottom (from left to right), the symbolic rules $r_2 = (L_2 \overset{l_2}{\leftarrow} K_2 \overset{r_2}{\rightarrow} R_2,\Phi_2)$ and $r_1 = (L_1 \overset{l_1}{\leftarrow} K_1 \overset{r_1}{\rightarrow} R_1,\Phi_1)$ are shown. In order to check direct confluence, both rules have to be applied to $SP_1$ and $SP_2$. However, this is not possible. If we want to find a symbolic match $o'_2: (L_2,\Phi_2) \rightarrow SP_1$ from the left-hand side of rule $r_1$ defined by $(L_2,\Phi_2)$ to the symbolic graph $SP_1=(P_1,\Phi_{P_1})$, we have to map label node $x$ of $L_2$ to label node $x'$ of $SP_1$. Mapping $x'$ of $L_2$ to $SP_1$ introduces two problems. The first problem is that no mapping of the label node $x''$ of $L_2$ to a label node in $SP_1$ exists such that $\dalg \models (\Phi_{P_1} \Rightarrow o'_{2,\Phi}(\Phi_2))$. We can overcome this problem by assuming that $SP_1$ still includes an additional variable, not assigned to any node or edge and not appearing in the formula of $SP_1$. Generally, we assume from now on that a symbolic graph also contains an unlimited number of variables. Nevertheless, we have a second problem: we still cannot apply $r_2$ to $SP_1$ because $x'=x+1 \land x''=x+2$ does not imply $o'_{2,\Phi}(\Phi_2)$ which is $x''' = x'+2$, where $x'''$ is the new additional variable for mapping $x''$ of $L_2$ to $P_1$ (i.e., $m_\Phi(x'')=x'''$).


\end{exmp}

This problem in Example~\ref{exdc} can be solved by \emph{narrowing graph transformation} \cite{Lazy}. Instead of requiring that $\Phi_{P_1} \Rightarrow o'_{2,\Phi}(\Phi_2)$ holds before the transformation (as in the case of symbolic direct derivation), in the narrowing case, the transformation of the E-graph part is performed first and, afterwards, the satisfiability of $\Phi_{P_1} \land o'_{2,\Phi}(\Phi_1)$ is checked to ensure that the resulting symbolic graph has at least one instance.   

\begin{defi}[Narrowing Graph Transformation \cite{Lazy}]
\label{def: Narrowing graph transformation}
\normalfont
Given a symbolic graph $SG=(G,\Phi_{G})$, a symbolic graph transformation rule $r = (L \leftarrow K \rightarrow R, \Phi)$ and an E-graph morphism $m: L \rightarrow G$, 
the \emph{narrowing direct derivation} of the rule $r$ on $SG$ at match $m$, denoted as $SG \Rrightarrow_{r,m} SH$, 
leading to symbolic graph $SH=(H,\Phi_{H})$, is given by the (E-graph) double pushout diagram below:
\begin{figure}[H]
\centering
\includegraphics[width=0.27\textwidth]{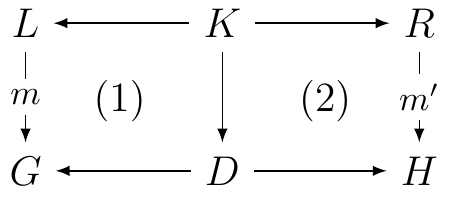}
\end{figure}
\noindent such that 
$\Phi_{H}:= \Phi_{G} \land m'_\Phi(\Phi)$ is satisfiable. 
\end{defi}

Now, we lift the notion of direct confluence to the rule level by using narrowing graph transformation.
\begin{defi}[\BNCP Pair]
\label{def: NCP}
\normalfont

A symbolic critical pair 
$SCP = SP_1 \overset{r_1,o_1}{\Longleftarrow} SK \overset{r_2,o_2}{\Longrightarrow} SP_2$  is a \emph{\NCP pair}
if there do not exist narrowing direct derivations $SP_1 \Rrightarrow_{r_2,o'_2} SX_1$ and $SP_2 \Rrightarrow_{r_1,o'_1} SX_2$ such that $SCP$ is directly confluent.


\end{defi}

Having these new concepts at hand, we can now revisit the concurrent applications of Example \ref{expardep} 
to see if a conflict detection based on \NCP pairs is now capable of handling that situation.

\begin{figure}[htbp]
\centering
\includegraphics[width=0.95\textwidth]{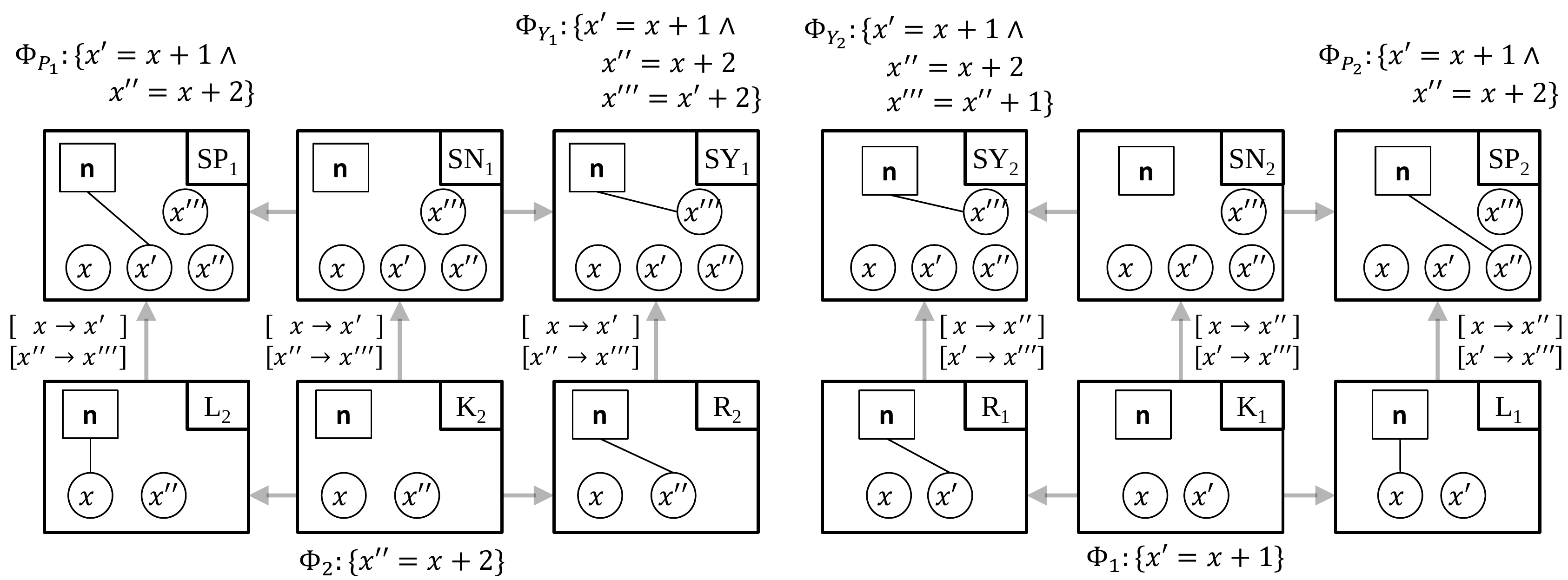}
\caption{Example of a Non-\NCP Pair}
\label{fig: Example NCP}
\end{figure}

\begin{exmp}[Non-conflicting Pair]\label{ex2}
\normalfont
Figure~\ref{fig: Example NCP} depicts the construction process for a \NCP pair according to Definition~\ref{def: NCP}, where $SP_1$ and $SP_2$ are part of the critical pair $SP_1 \overset{r_1,o_1}{\Longleftarrow} SK \overset{r_2,o_2}{\Longrightarrow} SP_2$ derived in Example~\ref{ex: Critical pair} (Figure~\ref{fig: Example of a critical pair}). Contrary to the previous example (Example~\ref{exdc}) the rules $r_1$ and $r_2$ (depicted at the bottom right and right of Figure~\ref{fig: Example NCP}, respectively) are now applied using narrowing transformation as defined in Definition~\ref{def: Narrowing graph transformation}. We also assume that symbolic graphs $SP_1$ and $SP_2$ both include a new label node $x'''$, which is used as image of the label nodes $x''$ and $x'$ in the (E-graph) matches $o'_2:L_2 \rightarrow P_1$ and $o'_1:L_1 \rightarrow P_2$, respectively. These mappings are depicted by the captions $[x'' \rightarrow x''']$ and $[x'\rightarrow x''']$ at the corresponding morphism arrows in Figure~\ref{fig: Example NCP}, respectively. The other mappings are depicted similarly, if the mapping differs from the mapping given by the node identifiers. 
The graphs $SY_1$ and $SY_2$ contain the results of the direct narrowing derivations of $r_1$ and $r_2$ at the matches $o'_1$ and $o'_2$. Consequently, the formula $\Phi_{Y_1}:=\Phi_{P_1}\land o'_{1,\Phi}(\Phi_2)$ can be simplified to $\Phi_{Y_1}:=\{x'=x+1 \land x''=x+2 \land x'''=x'+2\}$ as we have mapped $x$ to $x'$ and $x''$ to $x'''$. Having $\Phi_{Y_2}$ transformed similarly, we have $\Phi_{Y_1}:=\{x'''=x+1+2 \land x'' = x+2\}$ and $\Phi_{Y_2}:=\{x'''=x+2+1 \land x''=x+2\}$ which are equivalent. Hence, symbolic graphs $SY_1$ and $SY_2$ are isomorphic as both have the same graph structure and equivalent formulas.    

Concluding the example, direct confluence as a conflict condition can be used on the rule level as well, if we adapt the way how graph transformation is performed.
%
%
%
%

\end{exmp}
\section{An Improved Conflict Detection Process based on Direct Confluence}
\label{04_Algo}




%
%
%

The notion of \NCP pairs (Def.~\ref{def: NCP}) provides a basis for an improved conflict detection process. In this section, 
we describe this process. Thereupon, we show that the resulting set of \NCP pairs is \emph{complete} in the usual sense, i.e., 
whenever there is a conflict, we have a \NCP pair embedded in the input graph, which represents the cause 
of the conflict \cite{Fundamentals}.

A conflict detection based on \NCP pairs is not completely independent of 
a (classical) conflict detection based on critical pairs, but rather can be conceived as an extension to it. 
Such a conflict detection is performed on the rule level instead of the direct derivation level. 
Figure \ref{figdectree} summarizes the decision procedure.

In particular, given a pair of symbolic rules $r_1 = (L_1 \overset{l_1}{\leftarrow} K_1 \overset{r_1}{\rightarrow} R_1,\Phi_1)$ and 
$r_2 = (L_2 \overset{l_2}{\leftarrow} K_2 \overset{r_2}{\rightarrow} R_2,\Phi_2)$, the overall process consists of the following steps:
\begin{enumerate}[1.]
\item A symbolic critical pair (Def. \ref{def: scp}) is constructed if possible, based on $L_1,L_2$ and the matches. 
If the graph parts of $L_1$ and $L_2$ are non-overlapping, or $\mathcal{D} \not\models o_{1,\Phi}(\Phi_1) \wedge o_{2,\Phi}(\Phi_2)$ holds,
there is no \NCP pair based on these two rules 
and the process terminates. 
Note that, for the E-graph part, there is always at least one minimal graph according to Def. \ref{def: scp}.
\item If an appropriate $SK = (K,\Phi_K)$ with a minimal $K$ has been found in step 1, 
the direct derivations $SK \overset{r_1,o_1}{\Longrightarrow} SP_1$ and 
$SK \overset{r_2,o_2}{\Longrightarrow} SP_2$ (with the unique matches $o_1$ and $o_2$) are to be checked for parallel dependence. In case 
they are parallel independent, there is no \NCP pair based on these two rules 
and the process terminates.
\item The rules are applied in both sequences to $SK$; in case they are \emph{not} directly confluent, 
then $SK$, the rules $r_1$ and $r_2$ and their (unique) matches constitute a \emph{\NCP pair}.
\end{enumerate}

\begin{figure}[tp]
\centering
\includegraphics[width=0.73\textwidth]{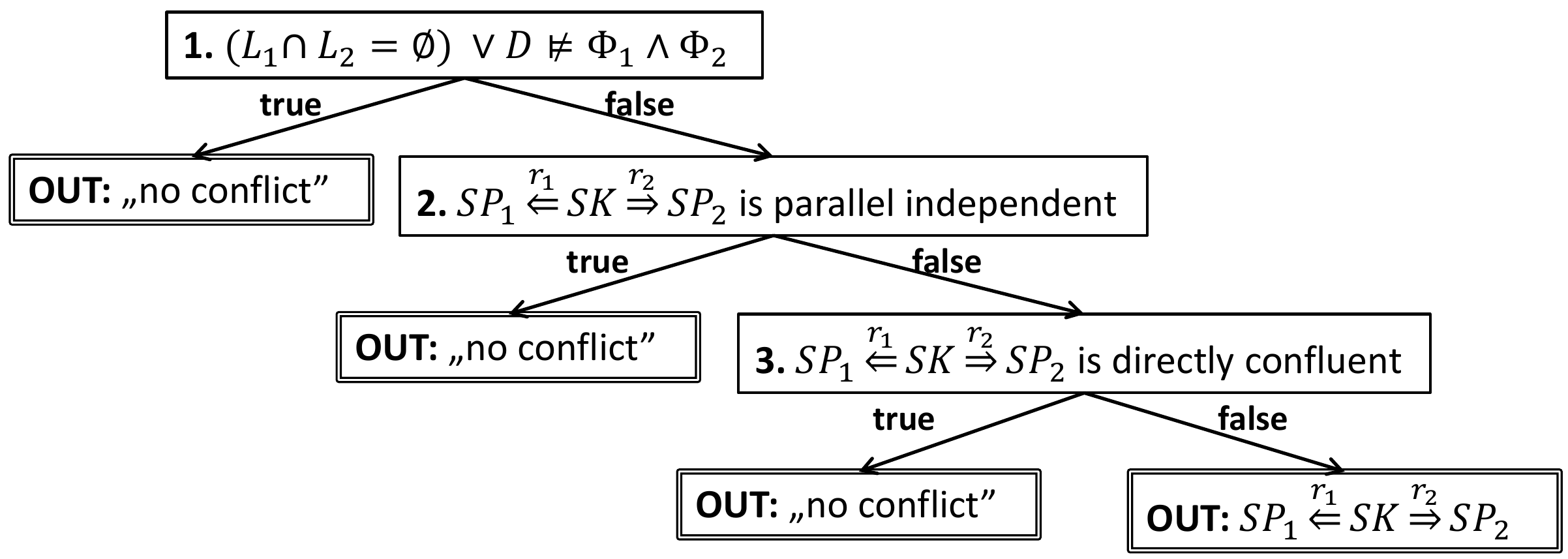}
\caption{Sketch of the Decision Procedure}
\label{figdectree}
\end{figure}

In the following, we prove that a conflict detection process defined this way is complete, i.e., 
when applied to a set of rules, the resulting set of \NCP pairs represents all possible conflict causes. 
This means that if for an arbitrary (symbolic) graph $SG$, two direct derivations are not directly confluent, 
then a corresponding \NCP pair is embedded within $SG$. In our proof, we rely on the construction of initial 
pushouts in symbolic graphs, analogously to the proof of Theorem~6.28 in \cite{Fundamentals}.
 \newpage
\begin{defi}[Construction of Initial Pushouts in Symbolic Graphs]\label{def: ip}
\normalfont

The diagram below is an \emph{initial pushout in symbolic graphs} if (i) the morphisms $b,c \in \mathcal{M}$,
(ii) it is an initial pushout in E-graphs (see Def.~6.1 in~\cite{Fundamentals}) and (iii) $\dalg \models (\Phi_B \Leftrightarrow \Phi_Y)$ and $\dalg \models (\Phi_C \Leftrightarrow \Phi_{X})$.

\begin{figure}[H]
\centering
\includegraphics[width=0.27\textwidth]{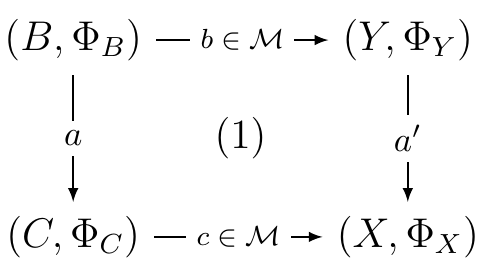}
\end{figure}

\end{defi}

\begin{thm}[Completeness of \BNCP Pairs]\label{thmcompleteness}

Given a grounded symbolic graph $\ul{SG}$ and a pair of \emph{not directly confluent} direct derivations $Der_G = (\ul{SH}_1 \overset{r_1,m_1}{\Longleftarrow} \ul{SG} \overset{r_2,m_2}{\Longrightarrow} \ul{SH}_2)$ of rules $r_1 = (L_1 \overset{l_1}{\leftarrow} K_1 \overset{r_1}{\rightarrow} R_1,\Phi_1)$ and 
$r_2 = (L_2 \overset{l_2}{\leftarrow} K_2 \overset{r_2}{\rightarrow} R_2,\Phi_2)$, there exists a \NCP pair 
$Der_K = (SP_1 \overset{r_1,o_1}{\Longleftarrow} SK \overset{r_2,o_2}{\Longrightarrow} SP_2)$
such that $Der_K$ can be embedded in $Der_G$ by $f: SK \to \ul{SG}$, $g: SP_1 \to \ul{SH}_1$ and $h: SP_2 \to \ul{SH}_2$ 
shown in the diagram:

\begin{figure}[H]
\centering
\includegraphics[width=0.40\textwidth]{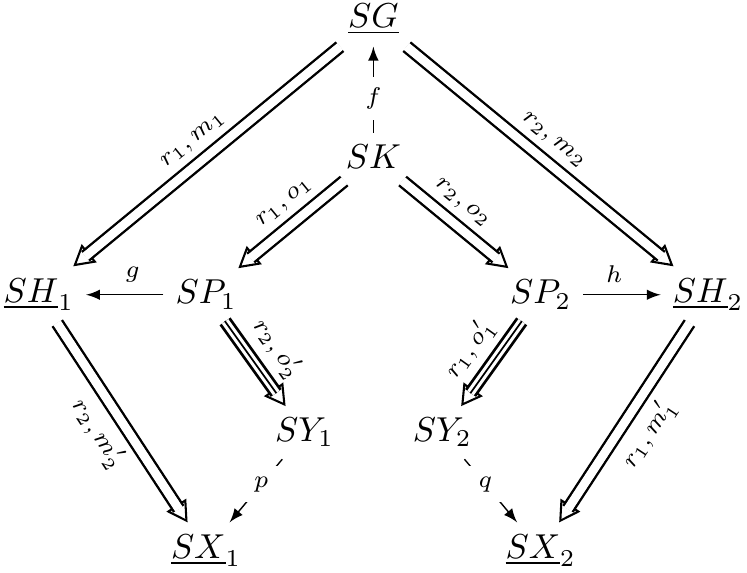}
\end{figure}

\end{thm}

\begin{proof}

First, we show that symbolic morphisms $f$, $g$ and $h$ exist.

As $Der_G$ 
is not directly confluent, it has to be parallel dependent. 
Due to the completeness of critical pairs (Lemma 6.22 in \cite{Fundamentals}), there exists a critical pair $Der_K$ in E-graphs with E-graph morphisms $f$, $g$ and $h$. Consequently, assuming that $Der_K$ is a symbolic critical pair (according to Def.~\ref{def: scp}), we have to show that $f$, $g$ and $h$ are symbolic graph morphisms.

Due to the existence of $Der_G$, we have $\dalg \models ((\Phi_G \Rightarrow m_{1,\Phi}(\Phi_1)) \land (\Phi_G \Rightarrow m_{2,\Phi}(\Phi_2)))$ which is equivalent to $\dalg \models (\Phi_G \Rightarrow m_{1,\Phi}(\Phi_1) \land m_{2,\Phi}(\Phi_2))$. From the minimality of critical pairs (i.e., $\mathcal{E'-M'}$ pair factorization \cite{Fundamentals}), it follows that $m_{1,\Phi} = f_\Phi \circ o_{1,\Phi}$ and $m_{2,\Phi} = f_\Phi \circ o_{2,\Phi}$, we have $(m_{1,\Phi}(\Phi_1) \land m_{2,\Phi}(\Phi_2)) \Leftrightarrow (f_\Phi(o_{1,\Phi}(\Phi_1)) \land f_\Phi(o_{2,\Phi}(\Phi_2)))$. By factoring out $f_\Phi$, we get $f_\Phi(o_{1,\Phi}(\Phi_1) \land o_{2,\Phi}(\Phi_2)) \Leftrightarrow f_\Phi(\Phi_K)$. Hence, $\dalg \models (\Phi_G \Rightarrow f_\Phi(\Phi_K))$ and, thus, $f$ is a symbolic graph morphism.

To show that $g$ and $h$ are symbolic graph morphisms, we require $(\Phi_{H_1} \Leftrightarrow \Phi_{H_2} \Leftrightarrow \Phi_{G})$ and ($\Phi_{P_1} \Leftrightarrow \Phi_{P_2} \Leftrightarrow \Phi_{K})$ as well as 
$(f_\Phi = g_\Phi = h_\Phi)$, which are consequences of Fact~\ref{fac: Properties of symbolic direct derivations}. If $\dalg \models (\Phi_G \Rightarrow f_\Phi(\Phi_K))$, also 
$\dalg \models (\Phi_{H_1} \Rightarrow g_\Phi(\Phi_{P_1}))$ and  
$\dalg \models (\Phi_{H_2} \Rightarrow h_\Phi(\Phi_{P_2}))$ and hence, $g$ and $h$ are symbolic graph morphisms.

We prove the rest of the theorem by contradiction. 
Let us suppose that there exist no symbolic direct derivations 
$\ul{SH}_1 \overset{r_2,m'_2}{\Longrightarrow} \ul{SX}_1$ and $\ul{SH}_2 \overset{r_1,m'_1}{\Longrightarrow} \ul{SX}_2$ 
with $\ul{SX}_1$ and $\ul{SX}_2$ being isomorphic, whereas, for the narrowing direct derivations 
$SP_1 \Rrightarrow_{r_2,o'_2} SY_1$ and $SP_2 \Rrightarrow_{r_1,o'_1} SY_2$, 
it holds that $SY_1$ and $SY_2$ are isomorphic.
In order to prove that this supposition is indeed a contradiction, it suffices to show that if $SY_1$ and $SY_2$ are isomorphic, then 
$\ul{SH}_1 \overset{r_2,m'_2}{\Longrightarrow} \ul{SX}_1$ and $\ul{SH}_2 \overset{r_1,m'_1}{\Longrightarrow} \ul{SX}_2$ 
exist, and $\ul{SX}_1$ and $\ul{SX}_2$ are isomorphic.


In the following, we rely on the technique used in the proof of the Local Confluence Theorem (Theorem 6.28 in \cite{Fundamentals}), which is based on initial pushouts. We adapt this procedure to our setting of symbolic graphs with $\mathcal{M}$-morphisms. Analogously to that proof, we first create an initial pushout over the morphism $f$ according to Def.~\ref{def: ip}. The pullback object $\ul{SZ}$, defined in Property~\ref{cond2} of direct confluence (Def.~\ref{def: dirconfl}) together with the closure property of initial pushouts (Lemma 6.5 in \cite{Fundamentals}) ensure that for each of the embedding morphisms, we have an initial pushout with $a: SB \to \ul{SC}$. 
The diagram below shows the last step of this construction. 
As symbolic graphs with $\mathcal{M}$-morphisms constitute an adhesive HLR category \cite{Fundamentals}, we only have to show that the results of the narrowing transformations are compatible with the construction of initial pushouts. 

\begin{figure}[H]
\centering
\includegraphics[width=0.9\textwidth]{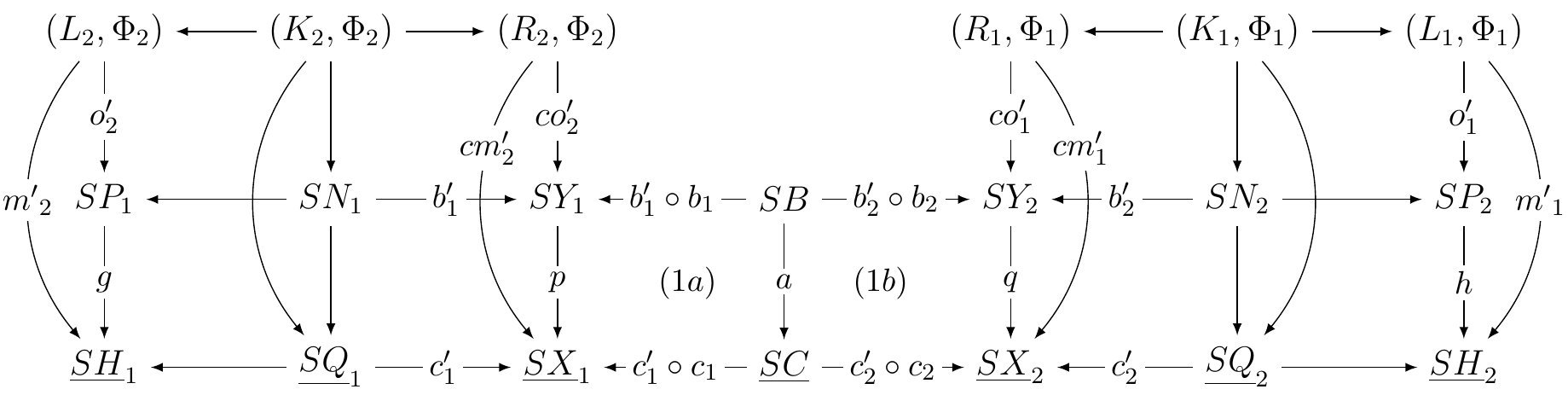}
\end{figure}

In particular, we have to show that if (2a) is an initial pushout in symbolic graphs, then (1a) is a pushout in symbolic graphs. 

\begin{figure}[H]
\centering
\includegraphics[width=0.4\textwidth]{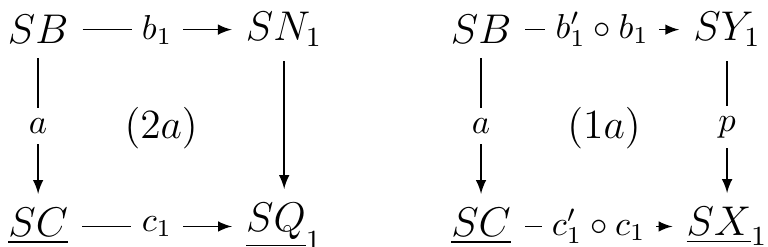}
\end{figure}
As (1a) is a pushout in E-graphs, this statement is equivalent to show that (i) morphisms $c'_1 \circ c_1$ and $b'_1 \circ b_1$ are symbolic graph morphisms and (ii) for the pushout (1a), $\dalg \models (\Phi_{X_1} \Leftrightarrow p_\Phi(\Phi_{Y_1}) \land c_{1,\Phi}(c'_{1,\Phi}(\Phi_C)))$ holds. 

(i). Since $c_1$ and $c'_1$ are both in $\mathcal{M}$, we can assume (without loss of generality) that $\Phi_C$ and $\Phi_{X_1}$ are the same formulas, and $V^D_C = V^D_{X_1}$ are the same sets of variables. Hence, $c'_{1,\Phi} \circ c_{1,\Phi}$ is the identity and, therefore, it is a symbolic graph morphism as $\dalg \models (\Phi_{X_1} \Rightarrow c'_{1,\Phi}(c_{1,\Phi}(\Phi_C)))$ trivially holds. For morphism  $b'_1 \circ b_1$, we have to show that $\dalg \models (\Phi_{Y_1}\Rightarrow b'_{1,\Phi}(b_{1,\Phi}(\Phi_B)))$ holds. By the definition of narrowing graph transformation, we have $\Phi_{Y_1}:=\Phi_{P_1} \land co'_{2,\Phi}(\Phi_2)$. It follows from the existence of the initial pushout $(\ul{SC} \leftarrow SB \rightarrow SP_1)$ that $\Phi_B\Leftrightarrow\Phi_{P_1}$ and hence, we have that $\dalg \models (\Phi_{Y_1}\Rightarrow b'_{1,\Phi}(b_{1,\Phi}(\Phi_B)))$ is equivalent to $\dalg \models ((\Phi_{P_1} \land co'_{2,\Phi}(\Phi_2)) \Rightarrow b'_{1,\Phi}(b_{1,\Phi}(\Phi_B)))$ which holds as $b_1$ and $b'_1$ are both in $\mathcal{M}$ and, therefore, $b'_{1,\Phi} \circ b_{1,\Phi}$ is the identity.

(ii). $\Rightarrow$: We have to show that $\dalg \models (\Phi_{X_1} \Rightarrow p_\Phi(\Phi_{Y_1}))$ and $\dalg \models (\Phi_{X_1} \Rightarrow c'_{1,\Phi}(c_{1,\Phi}(\Phi_C)))$ holds. While the latter has been already shown above, 
it remains to show that  
$\dalg \models (\Phi_{X_1} \Rightarrow p_\Phi(\Phi_{Y_1}))$. With $\Phi_{Y_1}:=\Phi_{P_1} \land co'_{2,\Phi}(\Phi_2)$ (from the definition of narrowing transformation), we have $(\Phi_{X_1} \Rightarrow p_\Phi(\Phi_{Y_1})) \Leftrightarrow (\Phi_{X_1} \Rightarrow p_\Phi(\Phi_{P_1} \land co'_{2,\Phi}(\Phi_2)))$ which is equivalent to $(\Phi_{X_1} \Rightarrow p_\Phi(\Phi_{P_1})) \land (\Phi_{X_1} \Rightarrow p_\Phi(co'_{2,\Phi}(\Phi_2))$. Due to Fact~\ref{fac: Properties of symbolic direct derivations}, we have $g_\Phi = p_\Phi$; by the construction of the symbolic direct derivation $\ul{SH}_1 \overset{r_2,m'_2}{\Longrightarrow} \ul{SX}_1$, $\Phi_{H_1} \Leftrightarrow \Phi_{X_1}$ holds; therefore, $\dalg \models (\Phi_{X_1} \Rightarrow p_\Phi(\Phi_{P_1}))$ is equivalent to
$\dalg \models (\Phi_{H_1} \Rightarrow g_\Phi(\Phi_{P_1}))$, which is given by the existence of the symbolic graph morphism $g:SP_1 \to \ul{SH}_1$.
It remains to show that $\dalg \models (\Phi_{X_1} \Rightarrow p_\Phi(co'_{2,\Phi}(\Phi_2)))$ which can be reformulated as 
$\Phi_{X_1} \Rightarrow cm'_{2,\Phi}(\Phi_2)$, using 
$p_\Phi \circ co'_{2,\Phi}=cm'_{2,\Phi}$. The implication $\dalg \models (\Phi_{X_1} \Rightarrow cm'_{2,\Phi}(\Phi_2))$ holds due to the existence of the symbolic graph morphism $cm'_2$.

(ii). $\Leftarrow$: By the construction of initial pushouts, we have that $c'_{1,\Phi}(c_{1,\Phi}(\Phi_C)) \Leftrightarrow \Phi_{X_1}$ and hence $p_\Phi(\Phi_{Y_1}) \land c'_{1,\Phi}(c_{1,\Phi}(\Phi_C)) \Rightarrow \Phi_{X_1}$.

We can show in the same way that (1b) is a pushout in symbolic graphs as well. It follows from the uniqueness of the pushout object that if $SY_1$ and $SY_2$ are isomorphic, so are $\ul{SX}_1$ and $\ul{SX}_2$.

This way, we have shown that our supposition contains a contradiction and, therefore, if $Der_G$ is not directly confluent, then $Der_K$ is a \NCP pair which can be embedded into $Der_G$.

\end{proof}

This proof shows that our proposed notion of \NCP pairs effectively represents the minimal conflict instances based on direct confluence and, thus, provides a means to lift conflict detection to rule level. Moreover, the general nature of the proof also demonstrates that the proposed technique is not restricted to the attributed setting used as motivation. In fact, direct confluence and \NCP pairs can be effectively used as an incremental extension of the existing conflict results for plain graphs as well.

\section{Related Work}
\label{05_RelWork}

\textbf{Symbolic graphs.} Symbolic graphs and symbolic graph transformation have been introduced by Orejas and Lambers in \cite{OL10a, Lazy} as a generalized 
and convenient representation for attributed graphs and attributed graph transformation. 
However, a proper notion of conflicts and a corresponding conflict detection process have not been considered 
in this framework.

\noindent\textbf{Conflicts.} 
The concept of conflicts has been adopted to graph transformation with negative application conditions and to attributed graph transformation with inheritance \cite{LaEhOr2006,GoLaEhOr2012}. In contrast to the proposed technique, these approaches rely 
on the notion of parallel dependence for determining conflicts. As a consequence, they still recognize 
a conflict whenever two rules access the same attribute and at least one modifies its value (regardless of the semantics of the access operations actually performed).
 
The concept of local confluence, which is a generalization of direct confluence, has its origins in term rewriting systems.  The applicability of local confluence to attributed graph transformation is shown in \cite{ConflOfAG}. However, in contrast to direct confluence, local confluence is undecidable even for graphs without attributes.
Additionally, the transformation of term attributed graphs, which is required to check local confluence, requires term unification to be performed at every derivation step. Contrary, in the symbolic case, where the formula is constructed stepwise at the syntactical level and is validated afterwards, e.g., by using off-the-shelf SMT solvers. 

\noindent\textbf{Refining conflict detection.} To the best of our knowledge, the only approach except for ours 
to formally capture and extend the notion of critical pairs is that of Lambers et al. \cite{EfficientConflictDetection}. 
They also try to narrow the set of actual conflicts, however, their approach is based on directly expressing the actual 
conflict cause by means of categorical notions and not on giving a new condition for checking which conflicts are considered relevant.

From a practical perspective, the approach of Cabot et al.~\cite{CabotSoSym2010} presents a fully-fledged graph transformation tool framework which also incorporates an analysis of graph transformation rules to verify certain properties, where their concept of conflict and independence strongly corresponds to our notion of direct confluence. The authors also remark that, similar to our technique, they only have to test the minimal models for those properties. Nevertheless, the approach of~\cite{CabotSoSym2010} is completely practical and it is based on a preceding translation of the rules into OCL expressions and, therefore, the theoretical aspects of our approach are not considered at all.


\section{Conclusion}
\label{06_Conclusion}

In this paper, we have proposed an improved conflict detection procedure for graph transformation with attributes. Our approach uses symbolic graphs as a framework and is based on the notion of conflicts and direct confluence. This way, we are able to explicitly take the intention of the attribute operations during conflict detection into account and to potentially exclude some false positive conflicts, emerging from the conservative conflict condition of earlier approaches, while still retaining completeness.

Based on this formal framework, we aim at implementing the approach using an off-the-shelf SMT solver, e.g., Z3, MathSAT or SMTInterpol~\cite{MoBj2008,mathsat5,ChHoeNu2012} and perform experiments regarding applicability and performance. Furthermore, we plan to apply this implementation to conduct case studies comprising modeling languages apparent in model-driven engineering.


\bibliographystyle{eptcs}
\bibliography{geza}

\end{document}